\documentclass[pra,twocolumn,english,superscriptaddress,floatfix,longbibliography]{revtex4-2}

\usepackage[english]{babel}
\usepackage{textcomp,xcolor,natbib}
\usepackage{dcolumn}
\usepackage{graphicx}
\graphicspath{{Figures/}}

\usepackage{amsmath, bbm}
\usepackage{latexsym}
\usepackage{amsfonts}   
\usepackage{amssymb}
\usepackage{array}      
\usepackage{epsfig}
\usepackage{txfonts}
\usepackage{xcolor,braket}
\usepackage[colorlinks=true,linkcolor=blue,urlcolor=blue,citecolor=blue]{hyperref}
\usepackage[normalem]{ulem}
\usepackage{soul}
\usepackage{dsfont}

\usepackage{xfrac}  
\usepackage{relsize}

\usepackage{amsthm}
\usepackage{mathrsfs, mathtools, amsmath}

\usepackage[normalem]{ulem}
\usepackage{cancel}
\usepackage{enumitem}

\usepackage{xspace, comment}

\newcommand{\ketbra}[1]{\ensuremath{\ket{#1}\bra{#1}}}
\newcommand{\id}{\ensuremath{\mathbbm{1}}}
\newcommand{\abs}[1]{\lvert #1 \rvert}

\newcommand{\tr}{\ensuremath{\operatorname{tr}}}
\renewcommand{\Re}{\operatorname{Re}}
\renewcommand{\Im}{\operatorname{Im}}

\definecolor{lavender}{rgb}{0.71, 0.49, 0.86}

\usepackage{enumitem}

\begin{document}
\title{Dynamics of Open Quantum Systems with Initial System-Environment Correlations via Stochastic Unravelings}

\author{Federico Settimo}
\email{fesett@utu.fi}
\affiliation{Department of Physics and Astronomy,
University of Turku, FI-20014 Turun yliopisto, Finland}

\author{Kimmo Luoma}
\affiliation{Department of Physics and Astronomy,
University of Turku, FI-20014 Turun yliopisto, Finland}

\author{Dariusz Chru\'sci\'nski}
\affiliation{Institute of Physics, Faculty of Physics, Astronomy and Informatics,
Nicolaus Copernicus University, Grudziadzka 5/7, 87-100 Toru\'{n},
Poland}

\author{Andrea Smirne}
\affiliation{Dipartimento di Fisica ``Aldo Pontremoli'', Universit{\`a} degli Studi di Milano, Via Celoria 16, I-20133 Milan, Italy}
\affiliation{Istituto Nazionale di Fisica Nucleare, Sezione di Milano, Via Celoria 16, I-20133 Milan, Italy}

\author{Bassano Vacchini}
\affiliation{Dipartimento di Fisica ``Aldo Pontremoli'', Universit{\`a} degli Studi di Milano, Via Celoria 16, I-20133 Milan, Italy}
\affiliation{Istituto Nazionale di Fisica Nucleare, Sezione di Milano, Via Celoria 16, I-20133 Milan, Italy}

\author{Jyrki Piilo}
\affiliation{Department of Physics and Astronomy,
University of Turku, FI-20014 Turun yliopisto, Finland}

\begin{abstract}
    In standard treatments of open quantum systems, the reduced dynamics is  described starting from the assumption that the system and the environment are initially uncorrelated.
    This assumption, however, is not always guaranteed in realistic scenarios and several theoretical approaches to characterize initially correlated dynamics have been introduced.
    For the uncorrelated scenario, stochastic unravelings are a powerful tool to simulate the dynamics.
    So far they have not been used in the most general case in which correlations are initially present{, since they cannot be applied to non-positive operators or non completely positive maps.}
    In our work, we employ the bath positive (B+) or one-sided positive decomposition (OPD) formalism as a starting point to generalize stochastic unraveling in the presence of initial correlations.
    Noticeably, our approach doesn't depend on the particular unraveling technique, but holds for both piecewise deterministic and diffusive unravelings.
    This generalization allows not only for more powerful simulations for the reduced dynamics, but also for a deeper theoretical understanding of open system dynamics.
\end{abstract}

\maketitle

\section{Introduction}
\label{sec:intro}
In standard open quantum system scenarios, one typically starts from the assumption that the system and the environment are {in a product state} at the initial time \cite{Breuer-Petruccione, Rivas-Huelga-OQS, Vacchini-OQS}
\begin{equation}
    \label{eq:factorized_init_state}
    \rho_{SE}(0) = \rho_S(0)\otimes\rho_E,
\end{equation}
with some fixed environmental state $\rho_E$.
If this is the case, the global system environment evolution is fixed by a unitary $U(t)$, while for the system is given by $\rho_S(t) = \Phi_t[\rho_S(0)]$, where \cite{Chruscinski-nM-local,Chruscinski2022}
\begin{equation}
    \label{eq:dynamics}
    \Phi_t[\cdot] = \tr_E[U(t)(\cdot\otimes\rho_E)U^\dagger(t)]
\end{equation}
is a completely positive trace-preserving (CPTP) map whose domain extends to the whole set of quantum states $\mathcal S(\mathscr H_S)$.
Assuming that $\Phi_t$ is invertible, one can introduce a two-parameters family of maps $\Phi_{t,s}\coloneqq \Phi_t\Phi_s^{-1}$, describing the evolution from time $s$ to time $t>s$.
If $\Phi_{t,s}$ is completely positive (CP) for all $t,s$, we say that the dynamics is CP divisible.
If $\Phi_{t,s}$ is only positive (P), then we say that it is P divisible.
Violations of both P and CP divisibility have been connected to quantum non-Markovianity \cite{BLP, BLP-PRA, RHP, rivas-quantum-nm, Wißmann-BLP-P-div, BLPV-colloquium, Budini2022}.

{The dynamical map $\Phi_t$ can be derived as the solution of the master equation} $d\rho/dt = \mathcal L_t[\rho]$, with the generator $\mathcal L_t = \dot\Phi_t\Phi_t^{-1}$ that can be put in Lindblad form \cite{Gorini1976,Lindblad1976}
\begin{equation}
    \label{eq:GKSL_ME}
    \mathcal L_t[\rho] = -i[H(t),\rho] + \sum_j\gamma_j(t)L_j(t)\rho L_j^\dagger(t)-\frac12\left\{\Gamma(t),\rho\right\},
\end{equation}
where $\Gamma(t) = \sum_j\gamma_j(t)L_j^\dagger(t)L_j(t)$.
From the point of view of the master equation, CP divisibility corresponds to the positivity of all rates $\gamma_i(t)\ge0$.
P divisibility, on the other hand, corresponds to the condition \cite{Kossakowski-necessary}
\begin{equation} \sum_j\gamma_j(t)\abs{\braket{\varphi_\mu\vert L_j(t)\vert\varphi_{\mu^\prime}}}^2\ge0
\end{equation}
for all orthonormal bases $\{\varphi_\mu\}_\mu$ and for all $\mu\ne\mu^\prime$.

However, the assumption of a {product} initial state, {although leading to a simple form of the master equation,} cannot cover all situations of interest {for instance when system and environment are strongly coupled}, so that it has been criticized from a fundamental and a practical perspective \cite{Pechukas1994, Alicki1995, Shaji2005, Carteret2008, Brodutch2013, Dominy2016}.
Different strategies have thus been devised both for the detection of initial correlations \cite{Dajka2010, Laine2010, Smirne2010, Zhang2010, Laine2012, HamedaniRaja2020} and their inclusion in the dynamical description \cite{Modi2012-operational-init-corr, Vacchini2016, Paz-Silva-B+, Colla2022}.
There is indeed no obvious and unique standpoint allowing to obtain the reduced dynamics of a system interacting with an environment if all these degrees of freedom are correlated.
In all cases, at variance with what happens in the framework of an initially factorized state, not all possible system states can appear as initial reduced system state.
A recently proposed strategy \cite{Colla2022}, that takes as starting point a fixed environmental state $\rho_E$ and a fixed correlation operator
$\chi$, admits as domain of possible initial reduced system states the statistical operators $\rho_S$ such that $\rho_S\otimes \rho_E+\chi$ is a positive operator.
The size of this convex set depends on the pair $\rho_E$ and $\chi$, and can even be reduced to a single point if environmental state and correlation operator are determined by a maximally entangled system-environment state.
For all elements in this set, the dynamics can be described by a time-local master equation of the form \eqref{eq:GKSL_ME}.
A different strategy \cite{Paz-Silva-B+} takes as starting point a specified correlated system-environment state, for which a suitable decomposition known as bath positive (B+) or one-sided positive decomposition (OPD) is given, which allows to obtain the time evolved reduced system state in terms of the action of a set of completely positive trace preserving maps of the form \eqref{eq:dynamics}, whose dimensionality is bounded by the square of the dimension of the system Hilbert space.
In this case the dynamics is no more described by a {single} time-local master equation of the form \eqref{eq:GKSL_ME}{, but by a set of master equations, one for each completely positive map}.
While the strategy works for a specified correlated system-environment state, suitable linear combinations of the maps introduced for this state allow to deal with all other initial states obtained by local transformations on the system degrees of freedom.

The master equation \eqref{eq:GKSL_ME} is typically very difficult to solve or simulate, and a powerful tool to deal with it is that of stochastic unravelings.
They consist of stochastic processes taking values on the system's Hilbert space, and the exact dynamics \eqref{eq:dynamics} is obtained by averaging over a large number of stochastic realizations{
\begin{equation}
    \label{eq:avg_unravelings}
    \rho(t) = \sum_{i}\frac{N_i(t)}N\ketbra{\psi_i(t)},
\end{equation}
where the $\ket{\psi_i(t)}$ are suitable random pure states.}
In different unraveling methods, the stochastic realizations can be of two major families: they can either consist of piecewise deterministic processes, interrupted by sudden jumps \cite{Plenio-jumps-review, Dalibard-MCWF, Diosi-orthogonal-jumps, PhysRevA.69.042107, Piilo-NMQJ-PRL, Piilo-NMQJ-PRA, PhysRevA.88.012124, Smirne-W, Chruscinski-Quantum-RO, Settimo-psi-ro, Settimo2025}, or they can be diffusive \cite{Percival-QSD, Gisin1992, Diosi-NMQSD, Caiaffa-W-diffusive, Luoma-diffusive-NMQJ}.
{In the following, we will consider different unraveling techniques, since our results are independent on the particular unraveling scheme used.
For the piecewise deterministic methods, we will consider the Monte-Carlo Wave Function (MCWF) \cite{Dalibard-MCWF}, which gives positive jump rates if and only if the dynamics is CP divisible, and the generalized rate operator (RO) \cite{Settimo-psi-ro, Settimo2025}, which can give positive rates also in some non-P divisible dynamics.
Both methods can be equipped with the non-Markovian Quantum Jumps (NMQJ) technique \cite{Piilo-NMQJ-PRL, Piilo-NMQJ-PRA} whenever the jump rates are temporarily negative.
For the diffusive methods, we will consider the Quantum State Diffusion (QSD) \cite{Gisin1992}, which, like the MCWF, can be applied only to CP divisible dynamics, although extensions to non-Markovian dynamics have been proposed \cite{Diosi-NMQSD}.
In Appendix \ref{app:unravelings} we provide an overview of these unraveling methods.}

So far, these stochastic methods have only been applied to scenarios in which the system and the environment are initially in a product state, since they are defined starting from the Lindblad master equation, which is guaranteed to hold only under this assumption.
In this work, we will present a method to apply unravelings techniques also in the case in which the dynamics cannot be obtained from a {single} master equation of the form \eqref{eq:GKSL_ME}, so as to allow the use of stochastic methods in combination with the OPD formalism for the treatment of initially correlated states of system and environment.

The rest of the paper proceeds as follows.
In Sec.~\ref{sec:OPD}, we recall two techniques for describing initially correlated dynamics.
Then, in Sec.~\ref{sec:unr_init_corr}, we present our main result, a simple way to apply the stochastic unravelings to the non-positive operators that are used in the OPD, with a simple example of the usefulness of our method.
In Sec.~\ref{sec:APO}, we recall the adaptive projection operator (APO) technique to derive the second order master equations corresponding to the CPTP maps of the OPD.
We also provide some examples of unravelings of initially correlated dynamics obtained via the APO technique.
Finally, in Sec.~\ref{sec:conclusions}, we present the conclusions of our work.

\section{Initially correlated system and environment}
\label{sec:OPD}
We now present in more detail the two techniques to deal with initially correlated system and environment states recalled in the introduction to better clarify the difference between the approaches and the point of connection with stochastic methods for the solution of the obtained reduced equations of motion.

\subsection{One-sided positive decomposition}
\label{subsec:OPD}
One possible way to describe the reduced dynamics of initially correlated system and environment is via the so-called OPD.
{This technique relies on the fact that any system side operator can be expanded in terms of frames}, i.e. possibly overcomplete basis of system side operators $\{Q_\alpha\}_\alpha\subset\mathcal L_2(\mathscr H_S)$ of self-adjoint Hilbert-Schmidt class operators \cite{Ali2000,Renes2004}.
{For arbitrarily dimensional systems and environments, one can always construct frames} such that any bipartite system-environment state $\rho_{SE}\in\mathcal S(\mathscr H_S\otimes\mathscr H_E)$ can be written as \cite{Paz-Silva-B+}
\begin{equation}
    \label{eq:OPD_state}
    \rho_{SE} = \sum_\alpha w_\alpha Q_\alpha\otimes\rho_\alpha,
\end{equation}
{where $w_\alpha$ are positive numbers, and the environmental operators $\rho_\alpha$ corresponding to the frame element $Q_\alpha$ are density matrices, i.e. environmental states.}
The representation of Eq.~\eqref{eq:OPD_state} is known as OPD.
{Such frame $\{Q_\alpha\}_\alpha$ giving states on the environmental side is highly non-unique, see \cite{Paz-Silva-B+} for more details.
In this work, we consider a special case of frame $\{Q_\alpha\}_\alpha$ such that, for $\alpha >0$, the frame elements $Q_\alpha$ are} proportional to the generalized Pauli matrices{, while $Q_0$ is} proportional to $\id_d-\sum_{\alpha>0} Q_\alpha$; {see Eq.~\eqref{eq:sys_operators_qubits} for an explicit construction of the frame for the case in which the system is a qubit.
The number of terms in the sum is bounded by $d^2$, where $d=\dim\mathscr H_S$ \cite{Smirne-OPD}.}

The reduced system state is obtained as
\begin{equation}
    \label{eq:OPD_reduced_state}
    \rho_S = \tr_E\rho_{SE}=\sum_\alpha w_\alpha Q_\alpha,
\end{equation}
with the weights and environmental states implicitly defined as
\begin{equation}
    \label{eq:env_state_OPD}
    w_\alpha\rho_\alpha = \tr_E[(P_\alpha\otimes\id_E)\rho_{SE}],
\end{equation}
where $\{P_\alpha\}_\alpha$ is the dual frame of $\{Q_\alpha\}_\alpha$, {i.e. the frame such that any operator $A$  acting on $\mathscr H_S$ can be written as
\begin{equation}
    A = \sum_\alpha \tr[A P_\alpha]Q_\alpha = \sum_\alpha \tr[A Q_\alpha]P_\alpha.
\end{equation}
Such dual frame} can be taken to be composed of positive operators such that $\tr[P_\alpha Q_\beta] = \delta_{\alpha \beta}$.
{In this case, it can be written as \cite{Paz-Silva-B+}
\begin{equation}
    \label{eq:dual_frame}
    P_\alpha = \sum_\beta M_{\alpha\beta}Q_\beta, \qquad M = (T^\top T)^{-1},
\end{equation}
where $T$ is a $d^2\times d^2$ matrix with coefficients $T_{\alpha\beta}=\tr[Q_\alpha G_\beta]$, and $\{G_\beta\}_\beta$ is an orthonormal basis of Hermitian operators on $\mathscr H_S$.}
{Notice that the weights and states of Eq.~\eqref{eq:env_state_OPD} depend not only on the frame $\{Q_\alpha\}_\alpha$ but also on the arbitrary global state $\rho_{SE}$.}
Furthermore, if $\rho_{SE}$ can be written using positive operators $Q_\alpha$ also on the system side, then it is separable \cite{Bengtsson2006}; if, additionally, $Q_\alpha$ or $\rho_\alpha$ (or both) are orthogonal projectors, then $\rho_{SE}$ is zero-discord \cite{Ollivier2001,Henderson2001,Modi2012}.

Eq.~\eqref{eq:OPD_state} allows us to write the time evolution of $\rho_S$ as the weighted sum of CPTP maps
\begin{equation}
    \label{eq:OPD_time_evolved}
    \rho_S(t) = \sum_\alpha w_\alpha\Phi^\alpha_t\left[Q_\alpha\right],
\end{equation}
{where the dynamical maps}
\begin{equation}
    \label{eq:OPD_maps}
    \Phi^\alpha_t[\cdot]\coloneqq\tr_E\left[U(t)(\cdot\otimes\rho_\alpha)U^\dagger(t)\right]
\end{equation}
are guaranteed to be CPTP since $\rho_\alpha\in\mathcal S(\mathscr H_E)${, and for each term of the sum system and environment are factorized.}


With the same set of CPTP maps $\{\Phi^\alpha_t\}$, it is possible to obtain not only the dynamics of $\rho_{SE}$ but also the dynamics of all states obtainable from it via system-side operations.
Let $\mathcal R$ be a completely-positive operation on $\mathcal S(\mathscr H_S)$ and
\begin{equation}
    \label{eq:system_repreparation}
    {\rho_{SE}^{\mathcal R}\coloneqq \frac{(\mathcal R\otimes\operatorname{id})\rho_{SE}}{\tr\left[(\mathcal R\otimes\operatorname{id})\rho_{SE}\right]},}
\end{equation}
then the time evolution of $\rho_S^{\mathcal R} = \tr_E\rho_{SE}^{\mathcal R}$ is given by
\begin{equation}
    \label{eq:OPD_repreparation}
    \rho_S^{\mathcal R}(t) = \frac{\sum_{\alpha,\alpha^\prime}w_\alpha R_{\alpha,\alpha^\prime}\Phi^\alpha_t\left[Q_{\alpha^\prime}\right]}{\sum_\alpha w_\alpha\tr\mathcal R\left[Q_\alpha\right]},
\end{equation}
where $R_{\alpha,\alpha^\prime}$ is an expansion of $\mathcal R$ in the basis $\{Q_\alpha\}$, i.e. $\mathcal R[Q_\alpha] = \sum_{\alpha,\alpha^\prime}R_{\alpha,\alpha^\prime}Q_{\alpha^\prime}$.
Therefore, with up to $d^2$ maps $\{\Phi^\alpha_t\}$, it is possible to describe the reduced dynamics starting from all system-environment states of the form \eqref{eq:system_repreparation}.
The set of {such global states} has dimension up to $(d^4-1)$, with the maximum obtained if $\rho_{SE}$ is maximally entangled \cite{Paz-Silva-B+}.
For instance, if the initial state is
\begin{equation}
    \label{eq:max_ent_state}
    \ket{\Psi_{SE}} = \frac1{\sqrt d}\sum_{k=0}^{d-1}\ket k\otimes\ket k,
\end{equation}
then all the maximally entangled generalized Bell states \cite{Bennett1993}
\begin{equation}
    \label{eq:gen_Bell_states}
    \ket{\Psi_{n,m}} = \frac1{\sqrt d}\sum_{k=0}^{d-1}e^{2\pi i k n/d}\ket{k\oplus m}\otimes\ket k
\end{equation}
with $n,m = 0,\ldots,d-1$, and $k\oplus m = k+m\mod d$, can be obtained via system side repreparations
\begin{equation}
    \ket k\mapsto e^{2\pi i k n/d}\ket{k\oplus m}.
\end{equation}
Also, all zero discord states $\rho^{\mathcal R_0}_S = \sum_{k=0}^{d-1}p_k\ketbra k\otimes\ketbra k$ can be obtained via the repreparation $\mathcal R_0[\rho] = \sum_{k=0}^{d-1}p_k\braket{k\vert\rho\vert k}\ketbra k$.
The price one has to pay in order to do so is to compute the $d^2\times d^2$ matrix $R_{\alpha,\alpha^\prime}$.
Nevertheless, computing this matrix is an arguably less complicated task than directly computing the time evolution of all states as in Eq.~\eqref{eq:system_repreparation}.

\subsection{Fixed correlations approach}
\label{subsec:single_ME_approach}
Another possible way to deal with initially correlated system and environment is to start by fixing the environmental state $\rho_E$ and correlations $\chi$, with $\chi=\chi^\dagger$, $\tr\chi=0$ \cite{Colla2022}.
The global state reads
\begin{equation}
    \label{eq:rho_SE_corr}
    \rho_{SE} = \rho_S\otimes\rho_E + \chi,
\end{equation}
and the reduced evolved state can then be written as
\begin{gather}
    \label{eq:dynamics_single_ME}
    \rho_S(t) = \Phi^\chi_t[\rho_S] \coloneqq \Phi_t[\rho_S] + I^\chi_t,\\
    \Phi_t = \tr_E\left[U(t)\rho_S\otimes\rho_EU^\dagger(t)\right],\quad
    I_t^\chi = \tr_E\left[U(t)\chi U^\dagger(t)\right],
\end{gather}
where $\Phi_t$ is the CP map describing the uncorrelated evolution, with the initial correlations entirely described by $I_t^\chi$.
From here, it is possible to derive a single master equation
\begin{equation}
    \label{eq:master_eq_single_ME}
    \mathcal L^\chi_t[X] \coloneqq \dot\Phi^\chi_t \circ \left(\Phi^\chi_t\right)^{-1}[X] = \mathcal L_t[X] + \Delta^\chi_t \tr X,
\end{equation}
where $\mathcal L_t = \dot\Phi_t\circ\Phi_t^{-1}$ is the generator of the uncorrelated dynamics, and
\begin{equation}
    \Delta^\chi_t = \dot I^\chi_t-\mathcal L_t[I^\chi_t]
\end{equation}
is the correlated part.
It is always possible to put $\Delta^\chi_t$ in Lindblad form
\begin{equation}
    \label{eq:Lindblad_form_correlated_ME}
    \Delta^\chi_t\tr X = \sum_i \eta_i(t)\left[J_i(t) X J_i^\dagger(t)-\frac12\left\{J_i^\dagger(t)J_i(t),X\right\}\right],
\end{equation}
where 
\begin{gather}
    J_i(t) = G_i(t) - \frac1d\tr[G_i(t)]\id_d,\\
    G_i(t) = \ket{\xi_j(t)}\bra{\xi_{j^\prime}(t)},\quad \eta_i(t) = b_j(t),
\end{gather}
$i=\{j,j^\prime\}$ is a double index and $\ket{\xi_j(t)}$ and $b_j(t)$ are (respectively) the eigenvectors and eigenvalues of $\Delta^\chi_t${, i.e. \begin{equation}
    \Delta^\chi_t = \sum_j b_j(t)\ketbra{\xi_j(t)}.
\end{equation}}
Since the master equation \eqref{eq:master_eq_single_ME} can always be put in Lindblad form, then {all} unraveling methods can be applied in a straightforward way.

The resulting master equation \eqref{eq:master_eq_single_ME}, however, {does not describe the evolution of the whole set of quantum states} $\mathcal S(\mathscr H)$, but only of a limited subset, depending on both $\rho_E$ and $\chi$ and containing all states $\rho\in\mathcal S(\mathscr H_S)$ such that
\begin{equation}
    \label{eq:physical_domain_single_ME}
    \rho\otimes\rho_E+\chi\ge0.
\end{equation}
This subset can even contain only a single state, as it happens for example if $\chi$ and $\rho_E$ are the correlations and environmental state corresponding to a maximally entangled system environment state.
For further details see Appendix \ref{app:comp_max_ent_max_mix}.

It is worth noticing that, since $\Delta_t^\chi$ is traceless, then the sum of the rates of the correlated generator $\eta_i(t)$ must also equal zero
\begin{equation}
    \sum_i\eta_i(t) {= \sum_j b_j(t)} = \tr\Delta^\chi_t = 0,
\end{equation}
and therefore at least one rate must be negative at all times.
Therefore, even if the uncorrelated generator $\mathcal L_t$ has all positive rates, it might happen that the addition of the correlated term gives rise to some negative rates for $\mathcal L^\chi_t$, thus making the dynamics more difficult and computationally expensive to unravel.

For example, if one considers $\chi$ and $\rho_E$ corresponding to a maximally entangled state and a global unitary evolution $U(t)$ giving a unital uncorrelated reduced evolution, i.e. $\Phi_t[\id]=\id$, then the corresponding dynamics has a negative rate already at $t=0$.
This happens because the generator at $t=0$ is simply
\begin{equation}
    \mathcal L_0^\chi[\rho_S(0)] = \Delta^\chi_0,
\end{equation}
since $\rho_S(0) = \id_d/d$ and $\mathcal L_t[\id] = 0$ because of unitality of $\Phi_t$.
In this case, not only the correlated generator $\Delta_t^\chi$ has some negative rates at $t=0$, but also the total generator $\mathcal L^\chi_t$ does. The fact that such negativity is present since $t=0$ not only makes the unravelings computationally more expensive, but might also cause the failure of the reverse jumps.
An explicit example of such kind of dynamics will be discussed in Sec.~\ref{subsec:dephasing}.

For these reasons, as well as the fact that unravelings can be {directly} applied in this formalism, in the following we will only focus on the OPD formalism, for which unravelings cannot be applied in a straightforward way since the initial condition of the master equation is not a density matrix.
{Notice that recently a Keldysh-contour based (two-state) unraveling that can deal with both uncorrelated and correlated initial global states has been introduced \cite{Cavina2025}.}


\section{Unravelings with One-sided positive decomposition}
\label{sec:unr_init_corr}

When system and environment are initially uncorrelated, it is possible to derive a single generator $\mathcal L_t$ in Lindblad form, valid for all open system initial states.
Starting from such a generator, it is possible to approximate the dynamics by applying the unraveling techniques to it.
So far, however, such stochastic techniques have only been applied to uncorrelated system and environment.
In this Section, we provide a way to apply them also to cases in which initial correlations are present, by exploiting the OPD formalism.

\subsection{Unravelings for non-positive operators}
\label{subsec:unr-non-pos}
\begin{figure}
    \centering
    \includegraphics[width=\linewidth]{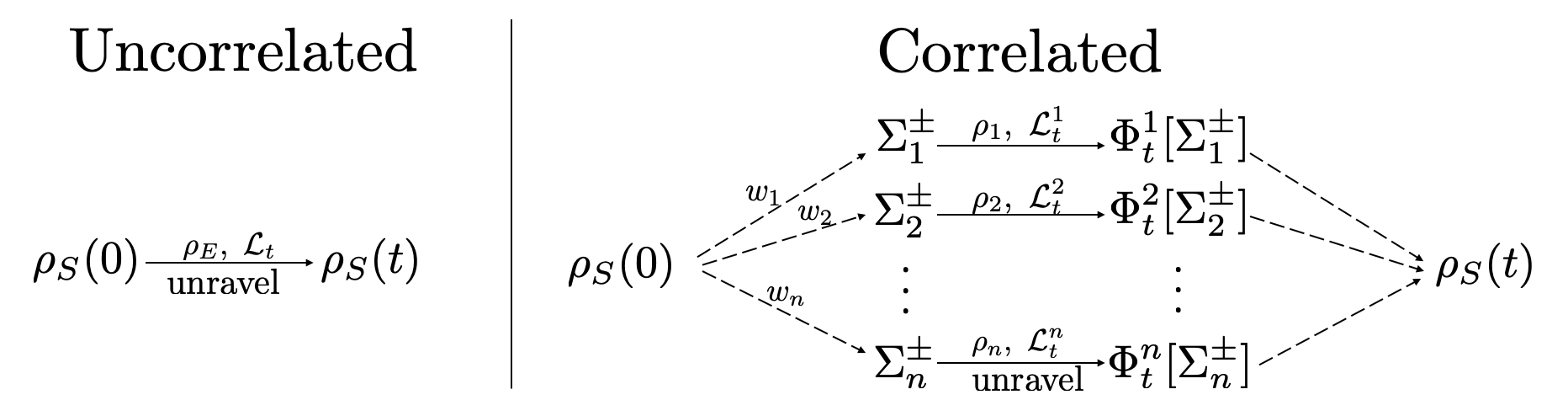}
    \caption{{For the uncorrelated case, a single Lindblad master equation, depending on $\rho_E$, is sufficient to describe the reduced evolution.
    If initial correlations are present, then one needs a set of up to $d^2-1$ master equations, one for each term of Eq.~\eqref{eq:OPD_state}, to describe the reduced evolution.
    Such master equations act on non-positive operators $Q_\alpha$, and therefore cannot be directly unraveled.
    However, as discussed in Sec.~\ref{sec:unr_init_corr}, they can be reconstructed from positive operators $\Sigma_\alpha^\pm$, which can be unraveled.}}
    \label{fig:OPD_repr}
\end{figure}

{For initially uncorrelated system environment, the reduced dynamics is described by a single Lindblad master equation $d\rho/dt = \mathcal L_t[\rho]$, depending on the environmental state $\rho_E$.
If, instead, system and environment are initially correlated, it is not possible to derive a single master equation.
However, the OPD \eqref{eq:OPD_state} allows one to write the correlated state $\rho_{SE}$ as a sum of uncorrelated terms $Q_\alpha\otimes\rho_\alpha$ with weights $w_\alpha$.
Since each term is uncorrelated and with a state $\rho_\alpha$ on the environmental side, the associated reduced operator will evolve according to a Lindblad master equation
\begin{equation}
    \label{eq:OPD_generators}
    \frac{d}{dt}Q_\alpha(t) = \mathcal L_t^\alpha\big[Q_\alpha(t)\big],\qquad
    \mathcal L^\alpha_t = \dot\Phi^\alpha_t\circ\Big(\Phi^\alpha_t\Big)^{-1},
\end{equation}
where the generator $\mathcal L^\alpha_t$ depends on the corresponding environmental state $\rho_\alpha$.
Since Eq.~\eqref{eq:OPD_state} contains up to $d^2-1$ terms, then to describe the reduced dynamics one needs up to $d^2-1$ generators $\{\mathcal L^\alpha_t\}_{\alpha=1}^{d^2-1}$, see Fig.~\ref{fig:OPD_repr}.
The maps $\Phi^\alpha_t$ of Eq.~\eqref{eq:OPD_maps} are the solutions of the Lindblad master equations generated by $\mathcal L^\alpha_t$.}

{When applying unravelings, another crucial difference between the two scenarios is that the master equation for the initially uncorrelated case has a state $\rho(0)$ as initial condition, and therefore its time evolution can be described as a convex mixture of pure states as in Eq.~\eqref{eq:avg_unravelings}.
When correlations are present, instead, each generator $\mathcal L^\alpha_t$ has a non-positive (but self-adjoint) operator $Q_{\alpha}$ as initial condition, which cannot be written as a convex mixture of pure states.}
Nevertheless, it is possible to reconcile stochastic unravelings and OPD by noticing that the system-side operators $Q_{\alpha}$ can always be written as the weighted difference of two states
\begin{equation}
    \label{eq:Q_pos_neg_part}
    Q_\alpha = Q_\alpha^+-Q_\alpha^- = \mu_\alpha^+\Sigma_\alpha^+-\mu_\alpha^-\Sigma_\alpha^-,
\end{equation}
where $Q_\alpha^\pm\ge0$ are the positive and negative part of $Q_\alpha$
\begin{equation}
    \label{eq:pos_neg_part_definition}
    Q_\alpha^\pm = \frac12\left(\abs{Q_\alpha}\pm Q_\alpha\right),\quad \abs X = \sqrt{X^\dagger X},
\end{equation}
$\mu_\alpha^\pm = \tr Q_\alpha^\pm$, and $\Sigma_\alpha^\pm = Q_\alpha^\pm/\mu_\alpha^\pm$.
It is straightforward to verify that the $\Sigma_\alpha^\pm$ defined this way are indeed self-adjoint, positive, trace-one operators, i.e. $\Sigma_\alpha^\pm\in\mathcal S(\mathscr H_S)$.
Furthermore, $\Sigma_\alpha^\pm$ do not depend on $\rho_{SE}$, but only on the fixed operators $Q_\alpha$.

Then it is possible to unravel the $\mathcal L_\alpha$ with $\Sigma_\alpha^\pm$ as initial conditions in the same way as for the uncorrelated case{, since the states $\Sigma_\alpha^\pm$ can be written as convex mixture of pure states}, thus obtaining $\Phi^\alpha_t[\Sigma_\alpha^\pm]$.
Linearity and positivity of $\Phi^\alpha_t$ ensure that this way of writing $Q_\alpha$ is preserved by the time evolution
\begin{equation}
    \label{eq:Phi_pos_neg_part}
    Q_\alpha(t) = \Phi^\alpha_t[Q_\alpha] = \mu_\alpha^+\Phi^\alpha_t[\Sigma_\alpha^+] - \mu_\alpha^-\Phi^\alpha_t[\Sigma_\alpha^-].
\end{equation}
{Therefore, unraveling with the states $\Sigma^\pm_\alpha$ as initial conditions allows to reconstruct the time evolution of $Q_\alpha$.}
The evolution of the system state can then be obtained as
\begin{equation}
    \label{eq:rho(t)_unr}
    \rho_S(t) = \sum_\alpha w_\alpha \mu_\alpha^+\Phi^\alpha_t[\Sigma_\alpha^+] - w_\alpha\mu_\alpha^-\Phi^\alpha_t[\Sigma_\alpha^-].
\end{equation}
Therefore, it is possible to obtain the dynamics of the marginal state $\rho_S$ of the initially correlated global state by unraveling each of the generators $\mathcal L^\alpha_t$, with initial states $\Sigma_\alpha^+$ and $\Sigma_\alpha^-$, and then recombining them with weights $\pm w_\alpha\mu_\alpha^\pm$.

In the literature, unraveling methods relying on additional degrees of freedom \cite{Imamoglu-stochastic, Garraway-nonperturbative-decay, Garraway-decay-atom-coupled, Breuer-double-Hilbert, Breuer-trajectories-nM, Becker2023} or temporarily negative probabilities for the occupation of certain states \cite{Muratore-Ginanneschi-influence-martingale} have been introduced.
These methods can also be applied to the generators obtained via the OPD.

The use of unraveling techniques is consistent with the system-side repreparations of Eq.~\eqref{eq:OPD_repreparation}.
The time evolution of any state $\rho_S^{\mathcal R} = \tr_E\left[(\mathcal R\otimes\operatorname{id})\rho_{SE}\right]$,
obtained from $\rho_{SE}$ via a system-only operation $\mathcal R$, can be written as
\begin{equation}
    \label{eq:sys_repreparations_unraveling}
    \rho_S^{\mathcal R}(t) = \frac{\sum_{\alpha,\alpha^\prime}w_\alpha R_{\alpha,\alpha^\prime}\left(\mu_{\alpha^\prime}^+\Phi^\alpha_t\left[\Sigma_{\alpha^\prime}^+\right] - \mu_{\alpha^\prime}^-\Phi^\alpha_t\left[\Sigma_{\alpha^\prime}^-\right]\right)}{\sum_\alpha w_\alpha\tr\mathcal R\left[Q_\alpha\right]}.
\end{equation}
Therefore, from unraveling the master equations $\mathcal L^\alpha_t$ with initial states $\Sigma^\pm_{\alpha^\prime}$, one can obtain the reduced dynamics of all states obtainable via local repreparations on the system only from a reference initial system-environment state $\rho_{SE}$.

{Notice that our method can be applied to arbitrarily dimensional systems, with the number of master equations to be considered that scales quadratically with the dimension of the system's Hilbert space, as it follows from the OPD technique of Sec.~\ref{sec:OPD}.
If one considers infinite dimensional systems, then one must necessarily truncate the dimensionality in order to apply any computational method, but this fact is not peculiar to our method.
Furthermore, the dimensionality of the environment doesn't play any role, since the number of master equation and the frame $\{Q_\alpha\}$ only depend on the system.
Therefore, our method can be applied in a straightforward way to infinite-dimensional or structured environments.
}

{For the sake of example, we now provide {one possible choice for a frame} $\{Q_\alpha\}_\alpha$, as well as the environmental states $\rho_\alpha$ for the special} case in which the system is a qubit $\mathscr H_S=\mathbb C^2$, without making any assumption on the dimension of $\mathscr H_E$.
One possible choice of frame for the system is given by
\begin{equation}
    \label{eq:sys_operators_qubits}
    Q_0 = \frac{\id-\sum_i\sigma_i}2,\qquad Q_i = \frac12\sigma_i,
\end{equation}
where $\sigma_i$, $i=x,y,z$, are the Pauli matrices, leading to a positive frame for the environment
\begin{equation}
    \label{eq:env_states_qubit}
    \rho_0 = \tr_S\rho_{SE},\quad
    \rho_i = \frac1{w_i}\tr_S\left[((\id+\sigma_i)\otimes\id_E)\rho_{SE}\right]
\end{equation}
with weights
\begin{equation}
    w_0 = 1,\qquad w_i = \tr_{SE}\left[((\id+\sigma_i)\otimes\id_E)\rho_{SE}\right].
\end{equation}
{The corresponding dual frame $\{P_\alpha\}_\alpha$ of Eq.~\eqref{eq:dual_frame} reads
\begin{equation}
    P_0 = \id,\qquad P_i = \id + \sigma_i.
\end{equation}}
This choice of a system frame can be generalized in a straightforward way to higher dimensional systems by considering the generalized Pauli matrices instead of the $\sigma_i$.

From Eq.~\eqref{eq:sys_operators_qubits}, it is possible to obtain the states $\Sigma_\alpha^\pm$ and weights $\mu_\alpha^\pm$ of Eq.~\eqref{eq:Q_pos_neg_part} in a straightforward way.
For $\alpha=x,y,z$, one has $\mu^\pm_x = \mu^\pm_y = \mu^\pm_z = 1$ and the corresponding states are simply the positive and negative eigenstates of the Pauli matrices
\begin{gather}
    \sigma_x = \ketbra+-\ketbra-,\quad \sigma_y = \ketbra{+_y}-\ketbra{-_y},\\
    \sigma_z = \ketbra1-\ketbra0.
\end{gather}
For $\alpha = 0$, the expression of the states $\Sigma_0^\pm$ instead reads
\begin{gather}
	Q_0 = \mu_0^+\ketbra{\phi^+_0}-\mu_0^-\ketbra{\phi^-_0},\quad\mu_0^\pm = \frac{\sqrt3\pm1}2,\\
	\ket{\phi^\pm_0} = \frac{\pm\sqrt3-1}{2\sqrt{3\mp\sqrt3}}(i-1)\ket1 + \frac1{\sqrt{3\mp\sqrt3}}\ket0.
\end{gather}

The explicit form of the environmental states \eqref{eq:env_states_qubit}, and therefore of the generators $\mathcal L^\alpha_t$, will depend on the full system-environment state.
Let us suppose that the initial state is entangled and of the form
\begin{equation}
    \ket{\Psi_{SE}} = \frac{\ket0\otimes\ket{\psi_0}+\ket1\otimes\ket{\psi_1}}{\sqrt2},
\end{equation}
with $\ket{\psi_{0,1}}$ not necessarily orthogonal.
{Let us stress the fact that we are not making any assumptions on the dimensionality of the environment and it can even be infinite-dimensional.}
The environmental states of the OPD are
\begin{gather}
    \label{eq:env_states_ME_1}
    \rho_0 = \frac12\left(\ketbra{\psi_0}+\ketbra{\psi_1}\right),\quad
    \rho_x = \ketbra{\phi_x},\\
    \label{eq:env_states_ME_2}
    \rho_y = \ketbra{\phi_y},\quad\rho_z = \ketbra{\psi_1},
\end{gather}
with 
\begin{gather}
    \ket{\phi_x} = \frac{\ket{\psi_0}+\ket{\psi_1}}{\sqrt{N_x}},\quad N_x = 2\left(1+\Re\braket{\psi_0\vert\psi_1}\right),\\
    \ket{\phi_y} = \frac{\ket{\psi_0}+i\ket{\psi_1}}{\sqrt{N_y}},\quad N_y = 2\left(1-\Im\braket{\psi_0\vert\psi_1}\right),
\end{gather}
with weights
\begin{equation}
    w_0 = w_z = 1,\quad w_{x,y} = \frac{N_{x,y}}2.
\end{equation}

The dynamics of $\Phi^{0,z}_t$ can be unraveled by mixing the dynamics obtained from the generators obtained with $\ket{\psi_i}$ as initial states.
For $\Phi^{x,y}_t$ this is indeed no longer the case and one needs to simulate the dynamics with different initial environmental states.

\begin{figure*}
    \centering
    \includegraphics[width=\linewidth]{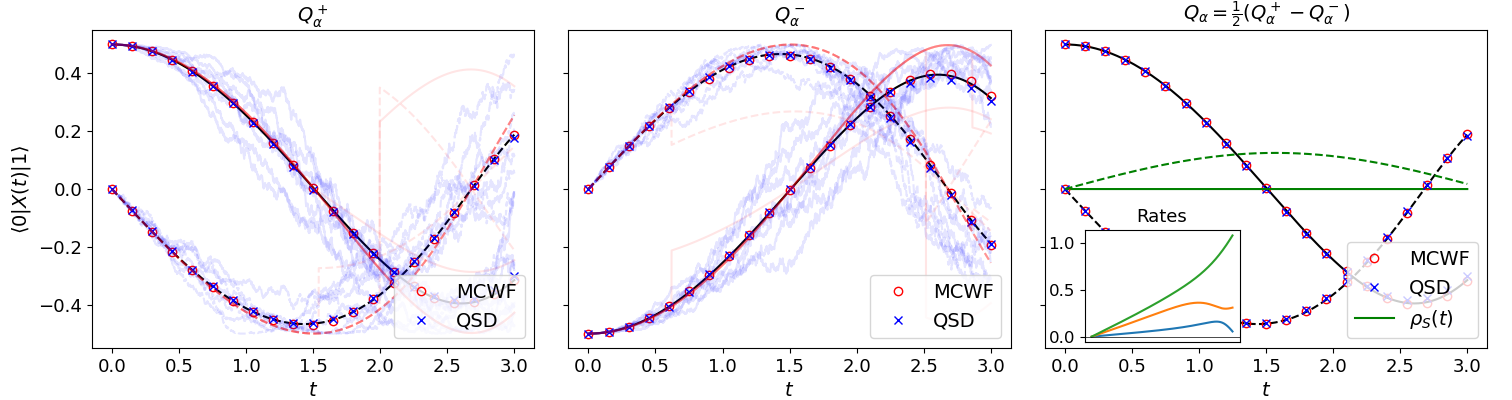}
    \caption{Unravelings of $\Phi^x_t(Q_x)$ for the dephasing dynamics of Eqs.~\eqref{eq:dephasing_Hamiltonian_d=4_free}-\eqref{eq:dephasing_Hamiltonian_d=4_coupling}; showing the real (solid) and imaginary (dashed) components of $\braket{0\vert \Phi^x_t(Q_x)\vert1}$.
    Left: unraveling of the positive part $\Sigma_x^+$.
    Middle: unraveling of the negative part $\Sigma_x^-$.
    Right: unraveling of $Q_x$, obtained as the difference between the two, and dynamics of the coherence of $\rho_S(t)$ (green).
    Inset: dephasing rates, i.e. eigenvalues of $K_{k\ell}^x(t)$.
    The unravelings are done using MCWF (red lines and circles) and QSD (blue lines and crosses) and for both 6 stochastic trajectories {(piecewise deterministic using MCWF and diffusive using QSD)} are shown in lighter shades.
    Parameters: $g=0.5$, $10^3$ stochastic realizations have been used.}
    \label{fig:dephasing}
\end{figure*}

\subsection{Example: dephasing dynamics}
\label{subsec:dephasing}

As a first example of the application of the unravelings to the OPD, let us consider an exactly solvable dephasing dynamics for a $d$-dimensional system \cite{Chruscinski2022}.
They are a class of dynamics such that the populations are preserved
\begin{equation}
    \frac d{dt}\braket{k\vert\Phi_t[\rho]\vert k} = 0\qquad\forall k=1,\ldots,d,
\end{equation}
while coherences are not.
This is obtained by considering system-environment Hamiltonians of the form
\begin{equation}
    \label{eq:dephasing_Hamiltonian}
    H = H_S\otimes\id_E + \id_S\otimes H_E + \sum_{k=1}^d \ketbra k\otimes B_k,
\end{equation}
where $H_S = \sum_k E_k\ketbra k$, and $B_k=B_k^\dagger$ are arbitrary environmental operators defining the coupling.
For this model, both the maps $\Phi^\alpha_t$ and the generators $\mathcal L^\alpha_t$ of Eqs.~\eqref{eq:OPD_maps}, \eqref{eq:OPD_generators} can be analytically calculated.
The master equations read
\begin{equation}
    \label{eq:dephasing_ME}
    \mathcal L^\alpha_t[X] = -i[H_\alpha(t),X] + \sum_{k,\ell=1}^{d-1}K_{k\ell}^\alpha(t)\left(S_k X S_\ell-\frac12\{S_\ell S_k,X\}\right),
\end{equation}
where
\begin{equation}
    \label{eq:S_dephasing}
    S_\ell = \frac1{\sqrt{\ell(\ell+1)}}\left(\sum_{k=0}^{\ell-1}\ketbra k - \ell \ketbra\ell\right),
\end{equation}
and $K_{k\ell}^\alpha(t)$ is a self-adjoint matrix which depends on the environmental state $\rho_\alpha$.
Eq.~\eqref{eq:dephasing_ME} can be put in Lindblad form by diagonalizing the matrix $K_{k\ell}^\alpha(t)$.
{This class of dephasing dynamics contains as a special case the spin boson dephasing in the presence of initial correlations \cite{Chaudhry2013, Mirza2021}.
Our method, being independent of the particular system-environment, can be applied also to this special case.}


{We now proceed to apply stochastic unravelings for this type of dynamics.}
For the sake of example, we now fix $d=4$, $\mathscr H_E = \mathscr H_S = \mathbb C^4$, and
\begin{gather}
    \label{eq:dephasing_Hamiltonian_d=4_free}
    H_S = H_E = \Omega\sum_{k=1}^4k\ketbra k\\
    B_k = g\Big(\ket k\bra{k+1}+\ket{k+1}\bra{k}\Big),\qquad k=1,2,3,
    \label{eq:dephasing_Hamiltonian_d=4_coupling}
\end{gather}
and $B_4=0$.
As initial state, let us consider the maximally entangled state
\begin{equation}
    \label{eq:init_state_dephasing}
    \ket{\Psi_{SE}} = \frac{\ket{1,1}+\ket{2,2}+\ket{3,3}+\ket{4,4}}{2}.
\end{equation}

The system-side frame $Q_\alpha$ can be obtained similarly to the qubit case of Eq.~\eqref{eq:sys_operators_qubits} as
\begin{equation}
    \label{eq:frame_d=4}
    Q_0 = \frac14\id_4-\frac12\sum_\alpha\sigma_\alpha,\qquad Q_\alpha = \frac12\sigma_\alpha,
\end{equation}
where the $\sigma_\alpha$ are the generalized $4\times4$ Pauli matrices.
The corresponding environmental states $\rho_\alpha$ can be found in Appendix \ref{app:OPD_d=4}.
{The unravelings can therefore be applied to the generators $\mathcal L^\alpha_t$ with initial conditions the system-side states $\Sigma^\pm_\alpha$ obtained from the positive and negative part of the operators $Q_\alpha$ of Eq.~\eqref{eq:frame_d=4} via Eq.~\eqref{eq:Q_pos_neg_part}.
The $Q_\alpha(t)$ are then reconstructed via Eq.~\eqref{eq:Phi_pos_neg_part} and the reduced dynamics via Eq.~\eqref{eq:rho(t)_unr}.}

With this choice, all the resulting maps $\Phi^\alpha_t$ are CP divisible up to some time, and therefore the unravelings can be performed using the MCWF and the QSD techniques up to that time.
The unraveling of $\mathcal L^x_t$, in the interaction picture with respect to the free Hamiltonian $H_S+H_E$, with initial conditions {the states $\Sigma_x^\pm$, corresponding to the frame element $Q_x$, are shown in the left and middle panel of Fig.~\ref{fig:dephasing}.
In the right panel, the dynamics of $Q_x$ is obtained from the dynamic of $\Sigma_x^+$ and $\Sigma_x^-$ by combining them with appropriate weights.}
The code used for obtaining the unravelings is available in \cite{github}.

{If one neglects the correlations and simply considers the dynamics of $\rho_S\otimes\rho_E$, then the resulting time evolution is trivial, since $\rho_S=\id/4$ and the dynamics is unital.
However, when initial correlations are taken into account, there is a revival in coherence $\braket{0\vert\rho_S(t)\vert1}$, as can be seen in the right panel of Fig.~\ref{fig:dephasing}.}

If one performs the unravelings of all $\mathcal L^\alpha_t$ with initial conditions $Q_{\alpha^\prime}$ then, via Eq.~\eqref{eq:OPD_repreparation}, it is possible to describe a subspace of initial states with dimension $d^4-1=255$.
This subspace contains, for example, all the generalized Bell states of Eq.~\eqref{eq:gen_Bell_states}, as well as all zero-discord states $\rho^{\mathcal R_0}_S = \sum_{k=0}^3p_k\ketbra k\otimes\ketbra k$.

It is worth stressing that if one describes the same dynamics with the fixed correlations approach of Sec.~\ref{subsec:single_ME_approach}, then the resulting dynamics is P indivisible since $t=0$.
This happens because the dephasing is unital and therefore the resulting master equation has at least one negative rate, as discussed in Sec.~\ref{subsec:single_ME_approach}, and therefore is CP indivisible.
But, since for the dephasing P and CP divisibility coincide  \cite{Chruscinski2022}, then it is also P indivisible since $t=0$.
This fact causes not only the MCWF to fail, but also the $\Psi$-RO.

\begin{figure*}
    \centering
    \includegraphics[width=\linewidth]{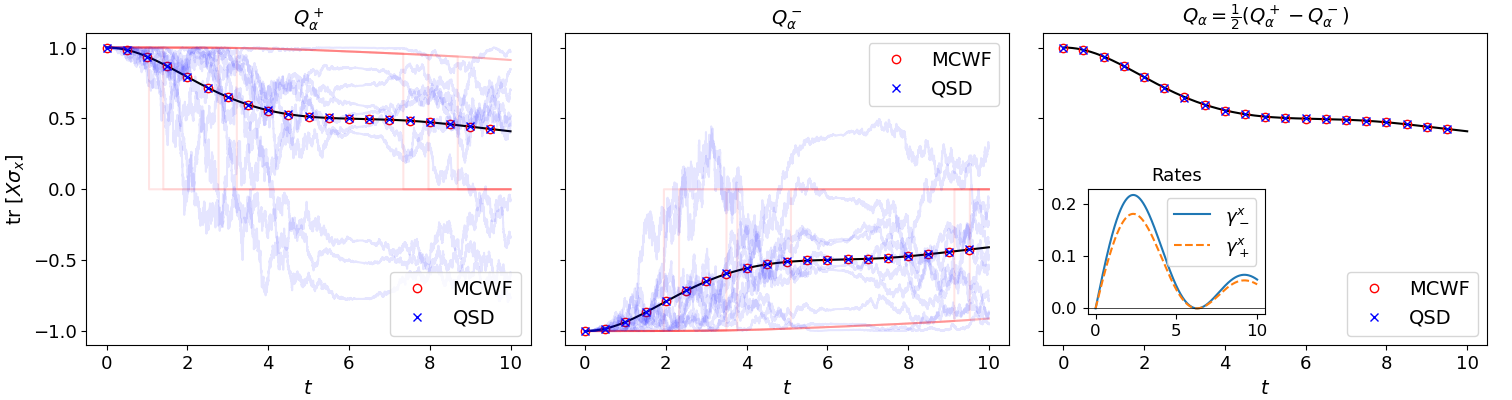}
    \caption{Unraveling of $\Phi^x_t[Q_x]$ for the Jaynes-Cummings dynamics in the continuum limit, for the maximally initial state of Eq.~\eqref{eq:JC_cont_init_state}.
    Left: unravelings for the positive part of $Q_x$; middle: for the negative part; right: $Q_x$ is obtained as the difference between the two.
    The unravelings are done both with the MCWF (red circles) and QSD (blue crosses), and 10 trajectories are shown.
    Inset: rates $\gamma^x_\pm$ for $\mathcal L^x_t$ of Eq.~\eqref{eq:ME_JC_APO}. Parameters: $g=0.05$, $N=10$.
    $10^4$ stochastic realizations have been used.}
    \label{fig:JC_APO_continuum}
\end{figure*}

\section{Unravelings with the adapted projection operator technique}
\label{sec:APO}
In order to apply stochastic unravelings to initially correlated system and environment, one does not need the a priori knowledge of the maps $\Phi^\alpha_t$, but the knowledge of the generators $\mathcal L^\alpha_t$ is sufficient.
{However, the derivation of the exact generators is in general a difficult task.}
In this Section, we recall the adapted projection operator (APO) technique, that allows one to derive second order master equations corresponding to the environmental states $\rho_\alpha$ of the OPD of Eq.~\eqref{eq:OPD_state}, thus drastically simplifying the task of deriving the generators.
We then apply this technique to two examples, performing the unravelings to the obtained master equations.

\subsection{Adapted projection operator technique}
\label{subsec:APO}
The APO technique generalizes the projector operator technique \cite{Breuer-Petruccione} by introducing projection operators of the form \cite{Trevisan-APO}
\begin{equation}
    \mathcal P_\alpha[\cdot] = \tr_E[\cdot]\otimes\rho_\alpha.
\end{equation}
Then, the dynamics of the relevant part $\mathcal P_\alpha[\rho_{SE}(t)]$ for the system side operators $Q_\alpha(t) = \Phi^\alpha_t[Q_\alpha]$ is described by the master equations $d Q_\alpha(t)/dt = \mathcal L^\alpha_t[Q_\alpha(t)]$, with
\begin{equation}
    \label{eq:ME_APO}
    \begin{split}
        \mathcal L^\alpha_t[X] =& -ig\sum_j\left[A_j(t),X\right]\braket{B_j(t)}_{\rho_\alpha}\\
        &-g^2\sum_{j,j^\prime}\int_0^td\tau\,\left[A_j(t),A_{j^\prime}(\tau)X\right]\operatorname{Cov}_{j,j^\prime}^{\rho_\alpha}(t,\tau)\\
        &+g^2\sum_{j,j^\prime}\int_0^td\tau\,\left[A_j(t),X A_{j^\prime}(\tau)\right]\operatorname{Cov}_{j^\prime,j}^{\rho_\alpha}(\tau,t),
    \end{split}
\end{equation}
where the operators $A_j(t)$ are the system-side operators of the interaction Hamiltonian in the interaction picture with respect to the free Hamiltonian
\begin{equation}
    H_I(t) = \sum_j A_j(t)\otimes B_j(t),
\end{equation}
with $A_j(t) = e^{iH_St}A_je^{-iH_St}$ and $B_j(t) = e^{iH_Et}B_je^{-iH_Et}$,
while
\begin{equation}
    \label{eq:APO_covariances}
    \operatorname{Cov}_{j,j^\prime}^{\rho_\alpha}(t,\tau) = \braket{B_j(t)B_{j^\prime}(\tau)}_{\rho_\alpha} - \braket{B_j(t)}_{\rho_\alpha}\braket{B_{j^\prime}(\tau)}_{\rho_\alpha},
\end{equation}
where we use the notation $\braket O_{\rho}= \tr[O \rho]$.

Therefore, by using the APO technique, it is possible to derive in a simple way a set of master equations, one for each component in the sum of Eq.~\eqref{eq:OPD_state}, describing separately the time evolution of each component $Q_\alpha$ of the reduced state \eqref{eq:OPD_reduced_state}.

\subsection{Example: Jaynes-Cummings}
\label{subsec:JC}
As a first example, let us consider a qubit $\mathscr H_S = \mathbb C^2$ interacting with a bosonic environment, with a Jaynes-Cummings form of the interaction
\begin{equation}
    \label{eq:H_JC}
    H = \frac{\omega_0}2\sigma_z+\sum_k\omega_kb_k^\dagger b_k + \sum_k\left(g_k \sigma_+\otimes b_k + g_k^+\sigma_-\otimes b_k^\dagger\right).
\end{equation}
{By applying the APO technique, the corresponding master equations read}
\begin{equation}
    \label{eq:ME_JC_APO}
    \begin{split}
        \mathcal L^\alpha_t[X] =& i \beta_+^\alpha(t)[\sigma_+\sigma_-, X]+i\beta_-^\alpha(t)[\sigma_-\sigma_+, X]\\
        &+\gamma_-^\alpha(t)\left(\sigma_- X \sigma_+ - \frac 12\{\sigma_+\sigma_-,X\}\right)\\
        &+\gamma_+^\alpha(t)\left(\sigma_+ X \sigma_- - \frac 12\{\sigma_-\sigma_+,X\}\right),
    \end{split}
\end{equation}
{with the rates $\gamma_\pm^\alpha$ and the driving $\beta_\pm^\alpha$ that can be found in \cite{Trevisan-APO}.}

\subsubsection{Continuum limit}
As a first example, let us consider an initial state of the form
\begin{equation}
    \label{eq:JC_cont_init_state}
    \ket{\Psi_{SE}} = \frac1{\sqrt 2}\ket0\otimes\ket0 + \frac1{\sqrt 2}\ket1\otimes\ket{\{N_k\}_k},
\end{equation}
where $\ket{\{N_k\}_k}$ is the environmental state with $N_k$ bosons in the mode of frequency $\omega_k$.
The explicit form of the weights $w_\alpha$ and environmental states $\rho_\alpha$ can be found in \cite{Trevisan-APO}.
{Although the APO technique was already used on this example, we stress that unravelings could not be performed because of the non-positivity of the system side operators.
Here, we use the same example and use our results of Sec.~\ref{sec:unr_init_corr} to apply unravelings to the resulting dynamics.}

We {consider infinitely many bath modes in the continuum limit} \cite{Breuer-Petruccione} and an Ohmic spectral density
\begin{equation}
    J(\omega) = g\omega\Theta(\omega_c-\omega),
\end{equation}
where $\Theta$ is the Heaviside theta function and $\omega_c$ a cutoff frequency, and we assume $N$ bosons up to $\omega_c$
\begin{equation}
    N(\omega) = N \Theta(\omega_c-\omega).
\end{equation}
We find that, under a suitable choice of the parameters, the rates $\gamma^\alpha_\pm$ of {the master equation \eqref{eq:ME_JC_APO}} remain positive at all times, as shown in the inset of Fig.~\ref{fig:JC_APO_continuum}.
Because of this, the unraveling methods for Markovian dynamics are sufficient for describing the dynamics.
Indeed, in Fig.~\ref{fig:JC_APO_continuum} we show the unravelings obtained using the MCWF and QSD techniques.
{Like in Sec.~\ref{subsec:dephasing}, we first unravel separately the positive and negative parts $\Sigma^\pm_\alpha$ of the frame element $Q_\alpha$ and then recombine the two to obtain the dynamics of $Q_\alpha$.}
From these unravelings, one is then able to also obtain the dynamics of all initial states obtainable via system-only repreparations according to Eq.~\eqref{eq:sys_repreparations_unraveling}.

\subsubsection{Single mode}

{Let us now consider the case of a single mode with a maximally entangled initial state of the form
\begin{equation}
    \label{eq:JC_single_mode_init_state-ent}
    \ket{\Psi_{SE}} = \frac{\ket0\otimes\ket{n_0} + \ket1\otimes\ket{n_1}}{\sqrt 2}.
\end{equation}
{For this intial state, the rates read}
\begin{gather}
    \gamma^\alpha_+(t) = \abs g^2(n_\alpha+1)\frac{\sin((\omega_0-\omega)t)}{\omega_0-\omega},\\
    \gamma^\alpha_-(t) = \abs g^2n_\alpha\frac{\sin((\omega_0-\omega)t)}{\omega_0-\omega},
\end{gather}
where $n_\alpha = \tr[\hat n\rho_\alpha]$, with $\rho_\alpha$ as in Eqs.~\eqref{eq:env_states_ME_1},~\eqref{eq:env_states_ME_2}.
It is worth noticing that $\Phi^0_t=\Phi^x_t=\Phi^y_t$, since the dynamics depends on $\rho_\alpha$ only via $n_\alpha$ and $n_0=n_x=n_y$.
It is easy to see that the rates are either both positive or both negative at the same time and correspondingly the dynamics is non-Markovian.}
Nevertheless, it is possible to perform the unravelings by using the NMQJ technique with an effective ensemble consisting of only three states: the deterministically evolving initial state $\ket{\psi_{\text{det}}(t)}$ and the eigenstates $\ket0$, $\ket1$ of $\sigma_z$.
The results are shown in the left panel of Fig.~\ref{fig:JC-single-mode-entangled}.

If one, additionally, performs the unravelings also of $\Phi^\alpha_t[Q_{\alpha^\prime}]$, $\alpha^\prime\ne\alpha$, then it is possible to obtain the evolution of all initial states obtainable from $\rho_{SE}$ via system only repreparations according to Eq.~\eqref{eq:sys_repreparations_unraveling}.
This task is drastically simplified by the fact that some maps are actually the same: $\Phi^0_t=\Phi^x_t=\Phi^y_t$.
{For instance, it is possible to recover the reduced dynamics corresponding to a zero discord initial state
\begin{equation}
    \label{eq:JC_single_mode_rep_0_disc}
    \rho_{SE}^{\mathcal R_0^p} = p\ketbra0\otimes\ketbra{n_0} + (1-p)\ketbra1\otimes\ketbra{n_1}
\end{equation}
or to the factorized initial state in which the correlations are eliminated
\begin{equation}
    \label{eq:JC_single_mode_rep_fact}
    \rho_{SE}^{\mathcal R_{\text{f}}} = \frac\id2\otimes\frac{\ketbra{n_0}+\ketbra{n_1}}2.
\end{equation}
The dynamics corresponding to such initial states is shown in the right panel of Fig.~\ref{fig:JC-single-mode-entangled}.
Clearly, the different kinds of correlations play non-trivial roles in modifying the resulting dynamics.}

\begin{figure}
    \centering
    \includegraphics[width=\linewidth]{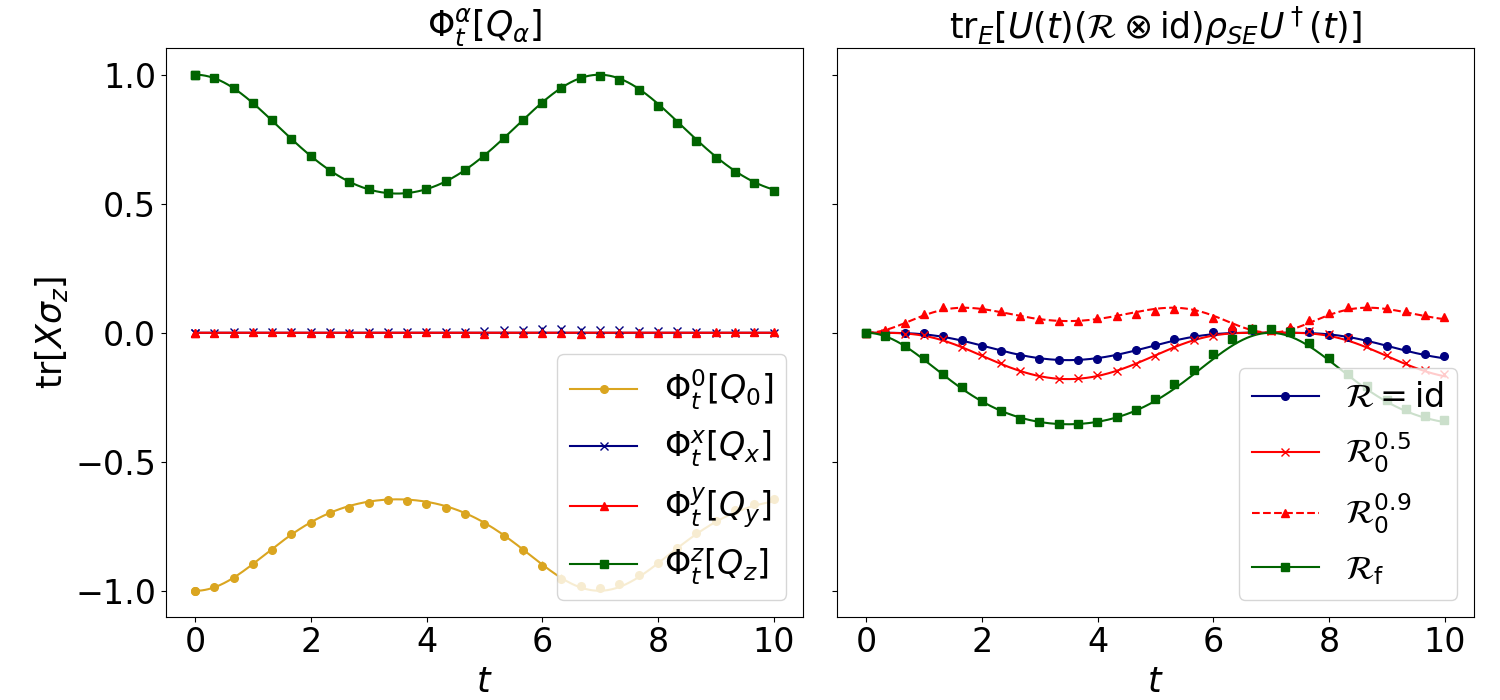}
    \caption{Unravelings of the single-mode Jaynes-Cummings dynamics for the maximally entangled initial state of Eq.~\eqref{eq:JC_single_mode_init_state-ent}, obtained using NMQJ.
    Left panel: $z$ component of the Bloch vector for $\Phi^\alpha_t[Q_\alpha]$.
    {Right panel: reduced dynamics for the maximally entangled state (blue) and of states obtained via the system-side repreparations: zero discord states $\rho_{SE}^{\mathcal R_0^p}$ of Eq.~\eqref{eq:JC_single_mode_rep_0_disc} with $p=0.5$ (red solid), $p=0.9$ (red dashed), factorized initial state $\rho_{SE}^{\mathcal R_{\text{f}}}$ of Eq.~\eqref{eq:JC_single_mode_rep_fact} (green).}
    Parameters: $n_0=1$, $n_1=0$, $\omega_0=1$, $\omega=0.1$, $g=0.5$.
    $10^4$ stochastic realizations have been used.}
    \label{fig:JC-single-mode-entangled}
\end{figure}

\subsection{Example: damped two-qubit model}
\label{subsec:two-qubits}

As a last example, let us consider two interacting qubits, considering one of them as our system and the other as the environment, with the environmental qubit also coupled to an external harmonic oscillator.
The global Hamiltonian reads
\begin{equation}
    \label{eq:H_sigma_xz}
    \begin{split}
        H =& \frac{\omega_1}2\sigma_z\otimes\id\otimes\id + \frac{\omega_2}2\id\otimes\sigma_z\otimes\id + \omega\id\otimes\id\otimes b^\dagger b\\
        &+ g \sigma_x\otimes\sigma_z\otimes\id + \mu\id\otimes\left(\sigma_+\otimes b+\sigma_-\otimes b^\dagger\right),
    \end{split}
\end{equation}
where $b^\dagger$ and $b$ are (respectively) the creation and annihilation operators.
As the initial state, we consider {a maximally entangled state, in which the entanglement is shared between the system qubit and the environmental qubit}, with the harmonic oscillator in the vacuum
\begin{equation}
    {\ket{\Psi_{SE}} = \frac{\ket{000}+\ket{110}}{\sqrt{2}}.}
\end{equation}
If one calculates the master equations with the APO formalism as in Eq.~\eqref{eq:ME_APO}, then the dynamics for $\rho_z$ {only adds a driving of the form}
\begin{equation}
    \label{eq:sigma_xz_ME_z}
    H_z(t) = g\left[\cos(\omega_1 t)\sigma_x-\sin(\omega_1t)\sigma_y\right].
\end{equation}
This happens because the $\rho_z$ covariance of Eq.~\eqref{eq:APO_covariances} is zero, and therefore no terms in $g^2$ appear in the master equation.
The other master equations, instead, {give non-trivial modifications to the Lindbladian
\begin{equation}
    \label{eq:sigma_xz_ME}
    \begin{split}
        \mathcal L^{0,x,y}_t[X] =& -i\left[H(t),X\right] + A(t)X\tilde A(t)\\ &+ \tilde A(t)XA(t)
        -\frac12\left\{\Gamma(t),X\right\},
    \end{split}
\end{equation}}
where
\begin{gather}
    H(t) = \frac{g}{\omega_1}[1-\cos(\omega_1t)]\sigma_z,\quad
    \Gamma(t) = 2\frac{g^2}{\omega_1}\sin(\omega_1t)\id,\\
    A(t) = g\left[\cos(\omega_1t)\sigma_x - \sin(\omega_1t)\sigma_y\right],\\ 
    \tilde A(t) = \int_0^td\tau\, A(\tau) = \frac g{\omega_1}\left[\sin(\omega_1t)\sigma_x - (1-\cos(\omega_1t))\sigma_y\right].
\end{gather}
{Notice that the master equations \eqref{eq:sigma_xz_ME} do not depend on the coupling $\mu$ with the external bath.
This happens because we assumed it to be in the vacuum, and the covariances in the APO formalism all have terms proportional to $\braket{0\vert b\vert 0}$ or $\braket{0\vert b^\dagger\vert 0}$, which are both zero.}

The jump term can be rewritten in the standard Lindblad form
\begin{equation}
    \label{eq:sigma_xz_jump_Lindblad_form}
    \mathcal J_t[X] = A(t)X\tilde A(t) + \tilde A(t)XA(t) = \sum_{i=\pm}\gamma_i(t)L_i(t)XL_i^\dagger(t),
\end{equation}
where $L_\pm(t)$ are a time-dependent combination of $\sigma_x$ and $\sigma_y$, with rates
\begin{equation}
    \label{eq:sigma_xz_rates}
    \gamma_\pm(t) = \sin(\omega_1t)\pm 2\sin\left(\frac{\omega_1t}2\right),
\end{equation}
that are shown in the inset of Fig.~\ref{fig:sigma_xz}.
Importantly, the resulting dynamics has a negative rate since $t=0$ and therefore methods such as the NMQJ fail, but it is possible to unravel the master equation \eqref{eq:sigma_xz_ME} using the $\Psi$-RO formalism.
Additionally, it is possible to have positive unravelings using only three states in the effective ensemble: $\ket{\psi_{\text{det}}(t)}$, $\ket0$, and $\ket1$.
These unravelings are shown in Fig.~\ref{fig:sigma_xz}.
Interestingly, in the $\Psi$-RO formalism, it is not necessary to explicitly compute the standard Lindblad form of Eq.~\eqref{eq:sigma_xz_jump_Lindblad_form}, but the form of Eq.~\eqref{eq:sigma_xz_ME} is enough.

\begin{figure}[t]
    \centering
    \includegraphics[width=\linewidth]{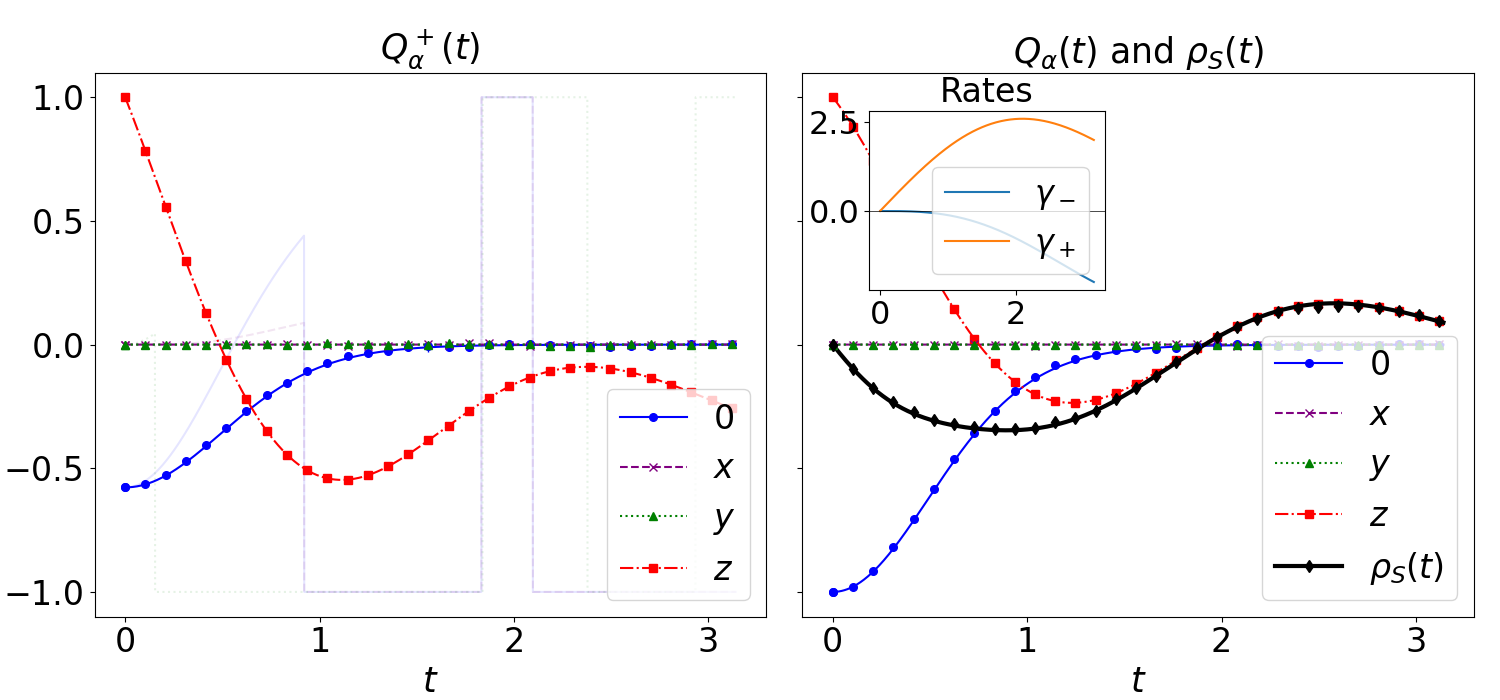}
    \caption{Unravelings of {the system qubit dynamics obtained from} Eq.~\eqref{eq:H_sigma_xz}.
    Left: dynamics of $\Phi^\alpha_t[Q_\alpha^+]$, showing (in lighter shade) 5 trajectories; for $t\le1$, no reverse jumps are needed even if $\gamma_-<0$.
    Right: dynamics of $\Phi^\alpha_t[Q_\alpha]$ and of $\rho_S(t)$.
    Inset: rates $\gamma_\pm$ of Eq.~\eqref{eq:sigma_xz_rates}.
    Parameters: $g=\omega_1=\omega_2 = \omega=\mu=1$.}
    \label{fig:sigma_xz}
\end{figure}

\section{Conclusions}
\label{sec:conclusions}
In this work we have extended the applicability of stochastic unraveling techniques to the most general case of open systems initially correlated with their environment.
As starting point, we considered the OPD formalism for describing the reduced dynamics, where the system density matrix is described by a sum of non-positive operators.
These operators pose challenges for unravelings, as they cannot be expressed as convex mixtures of pure states.
However, we addressed this limitation by decomposing each non-positive operator as the weighted difference of two states and performing the unravelings separately for each state.
Our approach enables the use of both piecewise deterministic and diffusive unravelings even in the most general case of initially correlated system and environment.

We validated our approach with examples, including dephasing and Jaynes-Cummings dynamics, with the master equations obtained either exactly or via the APO technique -- a generalizations of the projector operator technique in the OPD framework.
Different unraveling techniques were used, depending on the presence or absence of negative rates.
Our examples also demonstrated that a small number of unravelings suffice to describe the dynamics of all states obtainable from the initially correlated state via system-only operations.

We have also compared our method based on the OPD with another recently introduced formalism that starts from fixed correlations and allows one to obtain a single master equation valid on a subset of the system states.
We have shown that, if one starts from a maximally entangled state, the fixed correlation approach describes the dynamics of only a single state, while with the OPD one is able to characterize a $(d^4-1)$-dimensional subspace.
Furthermore, the fixed correlations master equation can present a negative rate since $t=0$, which makes the unravelings notoriously more complicated; the OPD formalism, on the other hand, avoids such issue, as shown explicitly in the dephasing example of Sec.~\ref{subsec:dephasing}.

Our results allow not only for the use of powerful simulation methods to the most general scenarios, but also for a deeper understanding of the reduced dynamics of initially correlated system and environment.

\section*{Acknowledgements}

F.S. acknowledges support from Magnus Ehrnroothin S\"a\"ati\"o.
A.S. and B.V. acknowledge support from MUR and Next Generation EU via the PRIN 2022 Project “Quantum Reservoir Computing (QuReCo)” (contract n. 2022FEXLYB)  and the NQSTI-Spoke1-BaC project QSynKrono (contract n. PE00000023-QuSynKrono).
D.C. was supported by the Polish National Science Center under Project No. 2018/30/A/ST2/00837.

\bibliography{biblio}

\appendix
\section{Stochastic unravelings of open system dynamics}
\label{app:unravelings}
{Here, we provide a small overview of some widely used unraveling methods, which are used in our work.}

\subsection{Monte-Carlo Wave Function}
\label{subsec:MCWF}
Whenever the dynamics is CP divisible (i.e. the rates $\gamma_j(t)$ in the master equation \eqref{eq:GKSL_ME} are positive at all times), it is possible to unravel the dynamics via the so-called Monte-Carlo Wave Function (MCFW) method \cite{Dalibard-MCWF}.
The deterministic evolution is given by
\begin{equation}
    \label{eq:MCWF_det}
    \ket{\psi(t)}\mapsto\ket{\psi(t+dt)} = \frac{(\id - iK(t)\,dt)\ket{\psi(t)}}{\lVert(\id - iK(t)\,dt)\ket{\psi(t)}\rVert}
\end{equation}
with the effective non-Hermitian Hamiltonian
\begin{equation}
    \label{eq:K_Hamiltonian}
    K(t) \coloneqq H(t) - \frac i2 \Gamma(t),
\end{equation}
and is interrupted by sudden jumps
\begin{equation}
    \label{eq:MCWF_jump}
    \ket{\psi(t)}\mapsto\ket{\psi(t+dt)} = \frac{L_j(t)\ket{\psi(t)}}{\lVert L_j(t)\ket{\psi(t)}\rVert}
\end{equation}
taking place with probability
\begin{equation}
    \label{eq:MCWF_jump_prob}
    p^j_{\psi(t)} = \gamma_j(t)\lVert L_j(t)\ket{\psi(t)}\rVert^2dt.
\end{equation}
Naturally, this method fails whenever $\gamma_j(t)<0$, since it would give rise to negative probabilities for the jumps.

\subsection{Non-Markovian Quantum Jumps}
\label{subsec:NMQJ}

Whenever the rates $\gamma_j(t)$ become temporarily negative, it is possible to unravel open system dynamics via the non-Markovian Quantum Jumps (NMQJ) technique \cite{Piilo-NMQJ-PRL,Piilo-NMQJ-PRA}.
The deterministic evolution is as in Eq.~\eqref{eq:MCWF_det}.
For the jumps, if the rate $\gamma_j(t)$ is positive, they are as in Eq.~\eqref{eq:MCWF_jump}.
If, on the other hand, $\gamma_j(t)<0$, then the jumps will be
\begin{equation}
    \label{eq:NMQJ_jump}
    \ket{\psi_i(t)} = \frac{L_j(t)\ket{\psi_{i^\prime}(t)}}{\lVert L_j(t)\ket{\psi_{i^\prime}(t)}\rVert} \mapsto \ket{\psi_{i^\prime}(t)}
\end{equation}
with probability
\begin{equation}
    \label{eq:NMQJ_jump_prob}
    p^j_{\psi_i(t)\to\psi_{i^\prime}(t)} = -\frac{N_{i^\prime}(t)}{N_i(t)}\gamma_j(t)\lVert L_j(t)\ket{\psi_{i^\prime}(t)}\rVert^2\,dt,
\end{equation}
where $N_{i,i^\prime}$ are the occupations of the realizations $\ket{\psi_{i,i^\prime}(t)}$, as in Eq.~\eqref{eq:avg_unravelings}.
These jumps are known as reverse jumps and are possible only if the source state $\ket{\psi_i(t)}$ is the possible target of a jump with a positive rate.
The reverse jump therefore has the effect of canceling a jump that had previously happened.
It is worth stressing that the probability of the reverse jump depends on the target state, instead of the source, and on the ratio $N_{i^\prime}(t)/N_i(t)$ and therefore the different stochastic realizations are not independent, thus making the simulations more expensive.

\subsection{Generalized Rate Operator}
\label{subsec:RO}
The master equation \eqref{eq:GKSL_ME} can be written as the sum of a jump term
\begin{equation}
    \label{eq:jump_ME}
    \mathcal J_t[\rho] \coloneqq \sum_j\gamma_j(t)L_j(t)\rho L_j^\dagger(t)
\end{equation}
and a driving term
\begin{equation}
    \label{eq:driving_ME}
    \mathcal D_t[\rho] \coloneqq - i(K(t)\rho - \rho K^\dagger(t)),
\end{equation}
where $K(t)$ is the non-Hermitian Hamiltonian of Eq.~\eqref{eq:K_Hamiltonian}.
Such a decomposition, however, is highly non-unique: any transformation \cite{Chruscinski-Quantum-RO}
\begin{gather}
    \label{eq:transf_J}
    \mathcal J_t[\rho]\mapsto\mathcal J_t^\prime[\rho] \coloneqq \mathcal J_t[\rho] + \frac12(C(t)\rho + \rho C^\dagger(t)),\\
    \label{eq:transf_K}
    K(t)\mapsto K^\prime(t) \coloneqq K(t)-\frac i2 C(t)
\end{gather}
preserves the structure of the master equation.

From this decomposition, it is possible to define an unraveling method with the deterministic evolution as in Eq.~\eqref{eq:MCWF_det}, with $K(t)$ substituted by $K^\prime(t)$, and jumps to the eigenstates $\ket{\varphi^i_{\psi(t)}}$ of the so-called rate operator (RO) \cite{Chruscinski-Quantum-RO,Diosi-orthogonal-jumps,Smirne-W}
\begin{equation}
    \label{eq:RO}
    R_{\psi(t)} \coloneqq \mathcal J^\prime[\ketbra{\psi(t)}],
\end{equation}
happening with probability
\begin{equation}
    \label{eq:RO_prob}
    p^i_{\psi(t)} = \lambda^i_{\psi(t)}\,dt,
\end{equation}
where $\lambda^i_{\psi(t)}$ is the eigenvalue corresponding to $\ket{\varphi^i_{\psi(t)}}$.
Note that both the eigenvalues and the eigenvectors depend on the pre-jump state.

In \cite{Settimo-psi-ro}, it was shown that, given that the stochastic realization is the state $\ket{\psi(t)}$, then it is possible to consider transformations $C(t)$ that depend not only on time, but also on the state of the realization $\ket{\psi(t)}$: $C(t)\mapsto C_{\psi(t)}$.
This leads to the introduction of the generalized RO ($\Psi$-RO), defined as
\begin{equation}
    \label{eq:psi-RO}
    \Psi\text{-}R_{\psi(t)} \coloneqq \mathcal J_t[\ketbra{\psi(t)}]+\frac12\left[\ket{\Phi_{\psi(t)}}\bra{\psi(t)} + \ket{\psi(t)}\bra{\Phi_{\psi(t)}}\right],
\end{equation}
with $\ket{\Phi_{\psi(t)}}\coloneqq C_{\psi(t)}\ket{\psi(t)}$.
The deterministic evolution is again as in Eq.~\eqref{eq:MCWF_det}, but with a non-linear effective Hamiltonian
\begin{equation}
    K_{\psi(t)} = H(t) - \frac i2\Gamma(t) - \frac i2 \ket{\Phi_{\psi(t)}}\bra{\psi(t)},
\end{equation}
and the jumps are to the eigenstates of $\Psi\text{-}R_{\psi(t)}$, with probability proportional to the corresponding eigenvalue.

With this new method, it is also possible to unravel some non-Markovian dynamics without the need to use reverse jumps.
Nevertheless, when this is not possible, both the RO and the $\Psi$-RO can be equipped with reverse jumps as the NMQJ technique of Sec.~\ref{subsec:NMQJ}.
We refer to unravelings in which only jumps with positive rates occur as positive unravelings.

\subsection{Quantum State Diffusion}
\label{subsec:QSD}
Quantum state diffusion (QSD) represents a well-known and widely used diffusive unraveling of open system dynamics. In QSD, {assuming that all rates are positive,} the state vector $\ket{\psi}$ obeys the stochastic Schr\"odinger equation \cite{Gisin1992}
\begin{equation}
    \label{eq:SSE}
    \begin{split}
        \ket{d\psi} =& -i H\ket{\psi}dt\\
        &+\sum_j\gamma_j\left(\braket{L_j^\dagger}_{\psi}L_j - L_j^\dagger L_j-\braket{L_j^\dagger}_{\psi}\braket{L_j}_{\psi}\right)\ket{\psi}dt\\
        &+\sum_j\sqrt{\gamma_j}\left(L_j-\braket{L_j}_\psi\right)\ket\psi dW_j,
    \end{split}
\end{equation}
where the explicit time dependence has been neglected, $\braket{L_j}_\psi = \braket{\psi\vert L_j\vert\psi}$, and $dW_j$ are independent complex Wiener processes satisfying
\begin{gather}
    \mathbb E[\Re(dW_j)\Re(dW_{j^\prime})] = \mathbb E[\Im(dW_j)\Im(dW_{j^\prime})] = \frac12\delta_{j,j^\prime}dt\\
    \mathbb E[dW_j] = \mathbb E[\Re(dW_j)\Im(dW_{j^\prime})] = 0,
\end{gather}
where $\mathbb E$ is the expectation value over the trajectories.

The stochastic Schr\"odinger equation \eqref{eq:SSE} is readily simulated by adding to the deterministic drift (the first two lines) the stochastic term, with the Wiener processes that are sampled by simply generating random complex Gaussian numbers with zero mean and standard deviation $\sqrt{dt}$.
The exact dynamics of Eq.~\eqref{eq:GKSL_ME} is then obtained by averaging over many realizations of the stochastic process.
As it happened for the MCWF, the positivity of all rates is required.
{However, it is possible to generalize the QSD to master equations having also negative rates via the non-Markovian QSD \cite{Diosi-NMQSD}.}

\section{Compatible system states with the fixed correlations approach}
\label{app:comp_max_ent_max_mix}

We now investigate the set of compatible states in the fixed correlations approach, as in Eq.~\eqref{eq:physical_domain_single_ME} for global states consisting of mixtures of the maximally entangled and the maximally mixed state.
Let $\dim\mathscr H_S=\dim\mathscr H_E=n$ and
\begin{equation}
    \label{eq:mix_max_ent_max_mix}
    \rho_{SE}(\lambda) = \lambda \ketbra{\Psi_{SE}} + (1-\lambda)\frac{\id_n}n\otimes\frac{\id_n}n,
\end{equation}
for $0\le\lambda\le1$, where
\begin{equation}
    \ket{\Psi_{SE}} = \frac1{\sqrt n}\sum_{i=1}^n\ket{\varphi_S^i}\otimes\ket{\varphi_E^i}
\end{equation}
is a maximally entangled state, where $\{\varphi_S^i\}_i$ and $\{\varphi_E^i\}_i$ are orthonormal sets for (respectively) the system and the environment.
The environmental marginal is the maximally mixed state, independently of $\lambda$, $\rho_E = \id_n/n$, while the correlations read
\begin{equation}
    \chi(\lambda) = \lambda\left[\ketbra{\Psi_{SE}}-\frac{\id_n}n\otimes\frac{\id_n}n\right].
\end{equation}

\newtheorem{proposition}{Proposition}
\begin{proposition}
    The set of compatible states $\rho_S(\lambda)$ such that $\rho_S(\lambda)\otimes\rho_E+\chi(\lambda)\ge0$ consists of all states of the form
    \begin{equation}
        \rho_S(\lambda) =\lambda \frac{\id_n}n + (1-\lambda)\sigma_S,
    \end{equation}
    where $\sigma_S$ is an arbitrary state in $\mathcal S(\mathscr H_S)$.
\end{proposition}
\begin{proof}
    In the basis $\{ \varphi_S^i \otimes \varphi_E^r \}_{i, r = 1,
\ldots, n}$, the operator $B (\lambda) \coloneqq \rho_S (\lambda) \otimes \rho_E + \chi (\lambda)$ reads
\begin{equation}
  B_{ir},_{js} (\lambda) = \left( \rho_S^{ij} (\lambda) -
  \frac{\lambda}{n} \delta_{ij} \right) \frac{1}{n} \delta_{rs} +
  \frac{\lambda}{n} \delta_{ir} \delta_{js},
\end{equation}
and is positive iff
\begin{enumerate}
  \item $B_{ir},_{ir} (\lambda) \geqslant 0$ \qquad$i, r = 1, \ldots, n$
  
  \item $| B_{ir},_{js} (\lambda) |^2 \leqslant B_{ir},_{ir} (\lambda)
  B_{js},_{js} (\lambda)$\qquad$i, r, j, s = 1, \ldots, n$ and $(i, r) \neq
  (j, s)$.
\end{enumerate}
These conditions become
\begin{enumerate}
    \item $\forall r = 1, \ldots, n$
    \begin{equation}
        \left( \rho_S^{ii} (\lambda) - \frac{\lambda}{n} \right) + \lambda
        \delta_{ir} \geqslant 0,
    \end{equation}
    and therefore
    \begin{equation}
          \rho_S^{ii} (\lambda) \geqslant \frac{\lambda}{n} \qquad \forall i =
          1, \ldots, n.
    \end{equation}
    It is convenient to consider the notation
    \begin{equation}
      \rho_S^{ii} (\lambda) = \frac{\lambda}{n} + q_i \qquad \forall i = 1,
      \ldots, n
    \end{equation}
    with $q_i \geqslant 0$ and $\sum_{i = 1}^n q_i = (1 - \lambda)$. For the
    initial maximally entangled state ($\lambda = 1$) the $q_i$ all vanish, while
    for the initial maximally mixed and hence factorized state ($\lambda = 0$) the
    $q_i$ are a probability distribution.
    \item For $(i, r) \neq (j, s)$, we have 
    \begin{equation}
      \left| \left( \rho_S^{ij} (\lambda) - \frac{\lambda}{n} \delta_{ij}
      \right) \delta_{rs} + \lambda \delta_{ir} \delta_{js} \right|^2 \leqslant
      (q_i + \lambda \delta_{ir}) (q_j + \lambda \delta_{js}),
    \end{equation}
    so that for $r \neq s$ the constraint is trivially satisfied
    \begin{equation}
      \lambda^2 \delta_{ir} \delta_{js} \leqslant (q_i + \lambda \delta_{ir})
      (q_j + \lambda \delta_{js}),
    \end{equation}
    while for $r = s$ and $i \neq j$ we have the non-trivial constraint
    \begin{equation}
      | \rho_S^{ij} (\lambda) |^2 \leqslant q_i q_j + \lambda (q_j
      \delta_{ir} + q_i \delta_{jr}) \qquad \forall r = 1, \ldots, n
    \end{equation}
    leading to
    \begin{equation}
      | \rho_S^{ij} (\lambda) |^2 \leqslant q_i q_j .
    \end{equation}
\end{enumerate}
We can thus write
\begin{equation}
  \rho_S (\lambda) = \lambda \frac{\id_n}{n} + Q
\end{equation}
with $Q$ a matrix with diagonal matrix elements fixed by the $q_i$, so that
$Q_{ii} \geqslant 0$ and $\sum_{i = 1}^n Q_{ii} = 1 - \lambda$, $|
Q_{ij} |^2 \leqslant Q_{ii} Q_{jj}$, that is a positive matrix with trace
equal to $(1 - \lambda)$. Apart from a multiplicative factor, such a matrix
is a statistical operator $Q = (1-\lambda)\sigma_S$.
\end{proof}

It is worth stressing that in the limit $\lambda = 0$, any state is compatible, since the global state is factorized.
For $\lambda=1$, instead, the global state is the maximally entangled state and the only compatible reduced system state is the maximally mixed state $\rho_S = \id_n/n$.

\section{OPD for $d=4$}
\label{app:OPD_d=4}

For the 4-dimensional example considered in Sec.~\ref{subsec:dephasing}, we considered a system side frame that generalises the one of Eq.~\eqref{eq:sys_operators_qubits} of $d=2$, i.e.
\begin{equation}
    Q_0 = \frac14\id_4 - \frac12\sum_\alpha\sigma_\alpha,\quad
    Q_\alpha = \frac12\sigma_\alpha,
\end{equation}
where $\sigma_\alpha$ are the generalized Pauli matrices
\begin{widetext}
\begin{gather}
    \sigma_1 = \begin{pmatrix}
        0&1&0&0\\
        1&0&0&0\\
        0&0&0&0\\
        0&0&0&0
    \end{pmatrix},\quad
    \sigma_2 = \begin{pmatrix}
        0&-i&0&0\\
        i&0&0&0\\
        0&0&0&0\\
        0&0&0&0
    \end{pmatrix},\quad
    \sigma_3 = \begin{pmatrix}
        1&0&0&0\\
        0&-1&0&0\\
        0&0&0&0\\
        0&0&0&0
    \end{pmatrix},\quad
    \sigma_4 = \begin{pmatrix}
        0&0&1&0\\
        0&0&0&0\\
        1&0&0&0\\
        0&0&0&0
    \end{pmatrix},\\
    \sigma_5 = \begin{pmatrix}
        0&0&-i&0\\
        0&0&0&0\\
        i&0&0&0\\
        0&0&0&0
    \end{pmatrix},\quad
    \sigma_6 = \begin{pmatrix}
        0&0&0&0\\
        0&0&1&0\\
        0&1&0&0\\
        0&0&0&0
    \end{pmatrix},\quad
    \sigma_7 = \begin{pmatrix}
        0&0&0&0\\
        0&0&-i&0\\
        0&i&0&0\\
        0&0&0&0
    \end{pmatrix},\quad
    \sigma_8 = \frac1{\sqrt3}\begin{pmatrix}
        1&0&0&0\\
        0&1&0&0\\
        0&0&-2&0\\
        0&0&0&0
    \end{pmatrix},\\
    \sigma_9 = \begin{pmatrix}
        0&0&0&1\\
        0&0&0&0\\
        0&0&0&0\\
        1&0&0&0
    \end{pmatrix},\quad
    \sigma_{10} = \begin{pmatrix}
        0&0&0&-i\\
        0&0&0&0\\
        0&0&0&0\\
        i&0&0&0
    \end{pmatrix},\quad
    \sigma_{11} = \begin{pmatrix}
        0&0&0&0\\
        0&0&0&1\\
        0&0&0&0\\
        0&1&0&0
    \end{pmatrix},\quad
    \sigma_{12} = \begin{pmatrix}
        0&0&0&0\\
        0&0&0&-i\\
        0&0&0&0\\
        0&i&0&0
    \end{pmatrix},\\
    \sigma_{13} = \begin{pmatrix}
        0&0&0&0\\
        0&0&0&0\\
        0&0&0&1\\
        0&0&1&0
    \end{pmatrix},\quad
    \sigma_{14} = \begin{pmatrix}
        0&0&0&0\\
        0&0&0&0\\
        0&0&0&-i\\
        0&0&i&0
    \end{pmatrix},\quad
    \sigma_{15} = \frac1{\sqrt6}\begin{pmatrix}
        1&0&0&0\\
        0&1&0&0\\
        0&0&1&0\\
        0&0&0&-3
    \end{pmatrix}.
\end{gather}
\end{widetext}
For the initial state
\begin{equation}
    \label{eq:init_state_dephasing_app}
    \ket{\Psi_{SE}} = \frac{\ket{0,0}+\ket{1,1}+\ket{2,2}+\ket{3,3}}{2}
\end{equation}
used in Eq.~\eqref{eq:init_state_dephasing}, the weights are
\begin{equation}
    w_\alpha = \tr\left[((\id_4+\sigma_\alpha)\otimes\id_4)\rho_{SE}\right]=1\qquad\forall\alpha.
\end{equation}
The environmental states
\begin{equation}
    \label{eq:rho_alpha_4d_app}
    \rho_\alpha = \tr_S\left[((\id_4+\sigma_\alpha)\otimes\id_4)\rho_{SE}\right]
\end{equation}
can be evaluated in a straightforward way, reading
\begin{equation}
    \rho_0 = \frac14\id_4,\qquad\rho_\alpha = \frac14\id_4+\sigma_\alpha.
\end{equation}

\end{document}